\documentclass[preprint,1p,12pt]{elsarticle}
\usepackage{fullpage}
\usepackage{amsfonts,amsmath,amssymb,amsthm,boxedminipage,color,url,fullpage}
\usepackage{enumerate,paralist}
\usepackage[numbers]{natbib} 
\usepackage[labelfont=bf]{caption}
\usepackage{aliascnt,cleveref}
\usepackage{graphics}
\usepackage{graphicx}
\usepackage{epic,eepic}

\newtheorem{defin}{Definition}[section]
\newtheorem{definition}{Definition}[section]

\newtheorem{theorem}[defin]{Theorem}

\newtheorem{coro}[defin]{Corollary}
\newtheorem{observation}[defin]{Observation}

\newtheorem{lemma}[defin]{Lemma}

\newcommand{\length}{\ell}
\newcommand{\dist}{\weightedDist}
\newcommand{\weight}{\omega}
\newcommand{\weightedDist}{\delta}

\newcommand{\WR}{\operatorname {wr }}
\newcommand{\BR}{\operatorname {BR }}
\newcommand{\WM}{\operatorname {wm }}
\newcommand{\BM}{\operatorname {BM }}
\newcommand{\LB}{\operatorname {LB }}
\newcommand{\BT}{\operatorname {BT }}
\newcommand{\HP}{\mathit{HP}}

\newcommand{\MST}{\operatorname{MST}}

\newcommand{\rad}{\operatorname{Rad}}
\newcommand{\graph}[1]{\left\langle #1\right\rangle}

\newcommand{\old}[1]{{}}

\def\A{{\cal A}}
\def\B{{\cal B}}

\def\R{{\mathbb R}}

\newcommand{\eqndef}{\stackrel{\mbox{\tiny def}}{{=}}}
\newcommand{\set}[1]{\left\{ #1 \right\}}
\newcommand{\size}[1]{\left| #1 \right|}

\def\cal{\mathcal}

\newcommand{\x}{x}
\newcommand{\lemref}[1]{Lemma~\ref{lem:#1}}

\newcommand{\thmref}[1]{Theorem~\ref{thm:#1}}
\newcommand{\figref}[1]{Figure~\ref{fig:#1}}

\journal{Theoretical Computer Science} 
\begin{document}
\begin{frontmatter}
\title{Optimizing Budget Allocation in Graphs}

\author{Boaz Ben-Moshe\corref{cor1}}
 \ead{benmo@g.ariel.ac.il}
 \address{Department of Computer Science and Mathematics, Ariel University, Ariel, Israel.}

\author{Michael Elkin\fnref{fn1}}
 \ead{elkinm@cs.bgu.ac.il}
 \address{Department of Computer Science, Ben-Gurion University of the Negev, Beer-Sheva, Israel.}

\author{Lee-Ad Gottlieb}
 \ead{leead@ariel.ac.il}
 \address{Department of Computer Science and Mathematics, Ariel University, Ariel, Israel.}

\author{Eran Omri}
 \ead{omrier@ariel.ac.il}
 \address{Department of Computer Science and Mathematics, Ariel University, Ariel, Israel.}

\cortext[cor1]{Corresponding author}
\fntext[fn1]{This author's research has been supported by the Binational Science Foundation, grant No. 2008390.}

\begin{abstract}
In a classical facility location problem we consider a graph $G$ with fixed weights on the edges of $G$. 
The goal is then to find an optimal positioning for a set of facilities on the graph with respect to some objective function.
We introduce a new framework for facility location problems, where the weights on the graph edges are not fixed, but rather 
should be assigned. The goal is to find the valid assignment for which the resulting weighted graph optimizes the 
facility location objective function.
%
 We present algorithms for finding the optimal {\em budget allocation} for the center point problem and for the median 
 point problem on trees. Our algorithms work in linear time, both for the case that a candidate vertex is given as part 
 of the input, and for the case where finding a vertex that optimizes the solution is part of the problem. 
 We also present a hardness result on the general case of the center point problem, followed by an $O(\log^2(n))$ 
 approximation algorithm for the general case - on general metric spaces.
\end{abstract}

\begin{keyword}
  Facility Location \sep Graph Radius \sep Budget Graphs
\end{keyword}

\end{frontmatter}

\section{Introduction}
A typical facility location problem has the following structure: the input includes a weighted set $D$ of demand locations, a set $F$ of feasible facility locations, and a distance function $d$ that measures the cost of travel between a pair of locations. For each $F' \subseteq F$, the quality of $F'$ is determined by some underlying
objective function ($obj$). The goal is to find a subset of facilities
$F' \subseteq F$, such that $obj(F')$ is optimized (maximized or minimized). 
One important class of facility location problems is the {\em center point}, 
in which the goal is to find one facility in $F$ that minimizes the maximum distance between
a demand point and the facility. Henceforth, we refer to this distance as {\em graph radius}.
In another important class of problems, {\em graph median},
 the goal is to find the facility in $F$ that minimizes the average distance (i.e., the sum of the distances) 
 between a demand point and the facility. 
In this paper we consider a new model for facility location on graphs, for which both problems are addressed.

\subsection{The New Model}
This paper\thanks{A preliminary version appeared as \cite{Ben-MosheEO11}} suggests a new model for budget allocation problems on weighted graphs.
The new model addresses optimization problems of allocating a fixed budget onto the graph edges where the goal is to find a subgraph that optimizes some objective function (e.g., minimizing the graph radius).
Problems such as center point and median point on trees and graphs have been studied extensively \cite{NDL_D_95,m-we1cp-83,LTUA_D_05,57954}. 
Yet, in most cases the input for such problems consists of a given (fixed) graph. 
Motivated by well-known budget optimization problems~\cite{1115584,ChepoiV02,OMCBDFBMDEAP-92,966446} raised in the context of communication networks, we consider the graph to be a communication graph, where the weight of each edge (link) corresponds to the delay time of transferring a (fixed length) message over the link. We suggest a {\em Quality of Service} model for which the weight of each edge in the graph depends on the budget assigned to it. In other words, paying more for a communication link decreases its delay time.

More formally, we consider an undirected graph $G=\graph{V,E}$ equipped with a length function $\length(\cdot)$ on the edges. The weighted graph 
$(G,\length)$ induces a metric space. i.e., the length function $\length(\cdot)$ satisfies the triangle inequality. Let $B$ be a positive value, which we call a {\em budget parameter}. In our model, allocating a budget 
$\B(e) \ge 0$ to an edge $e \in E$ of length $\length(e)$ means that the {\em weight}  of $e$ will be $w(e) = w_{\B}(e) = \frac{\length(e)}{\B(e)}$ (in case $\B(e) = 0$, we let $w(e)=+\infty$).

The distance between two vertices $v$ and $u$ in $(G,\B)$ is the sum of weights on the edges of a minimal weighted path between them, denoted as $\dist_{\B}(v,u)$.
The {\em radius} of $(G,w)$ with respect to a designated vertex $r$  
$\rad_r(G,w)$ is the maximum distance (with respect to the underlying allocation $\B(\cdot)$) between $r$ and any other vertex $v$ in the graph.
In the {\em rooted budget radius} problem, we are given a graph $G= \graph{V,E}$, equipped with a length function $\length(\cdot)$, a budget parameter $B$ and a designated vertex $r$. The goal is to assigned none-negative values $\B(e)$ to the edges of $G$ that sum up to $B$ (i.e., $\Sigma_{e \in E}(\B(e)) = B$),  and so that the rooted radius of the graph $(G,w)$ with $w=\frac{\length}{\B^{c}}$ (for some positive constant $c$; in this paper we mainly consider $c$ to be 1), with respect to the vertex $r$ is minimized.

In the {\em unrooted budget radius} problem, we are given a graph $G= \graph{V,E}$, equipped with a length function $\length(\cdot)$, and a budget parameter $B$. The goal is to minimize the rooted budget radius problem over all choices of $r\in V$. That is, our goal is to find a vertex $r$ and a corresponding assignment $\B(e)$ that minimizes the budget radius problem with respect to $G$, $B$, and $r$.
The \emph{budget diameter} problem is defined analogously, where the goal is to minimize the maximum distance between any pair of vertices in $(G,\weight)$.

In the {\em budget median} problem, we are given a graph $G= \graph{V,E}$, equipped with a length function $\length(\cdot)$, and a budget parameter $B$. The goal is to find a vertex $r$ and a corresponding assignment $\B(e)$ that minimize the \emph{average distance} between $r$ and other vertices of the graph, defined as $\frac{1}{n}\cdot\sum_{v\in V}{\dist_{\B}(r,v)}$. The vertex $r$ that minimizes the budget median problem is called the {\em budget median}.

\old{
Divide $B$ among the edges of $E$ in a manner that the accumulated weight of any simple path from $r$ to any other node (that is, the new radius) is minimized. Using this {\em weight} budget-function we can also define an unrooted version of the budget radius problem, i.e, finding the minimum solution over all the vertices in $G$. Other related problems are (a) {\em median point}: minimizing the graph average weight and
(b) graph {\em diameter}: minimizing the weight of the heaviest simple shortest path in the graph. 
Noteworthy is that the model presented applies to trees as well as to general (connected) graphs.
}
\subsection{Motivation}
We were motivated by communication optimization problems in which for a fixed 'budget' one 
needs to design the 'best' network layout. The quality of 
service ({\em QoS}) of a link between two nodes depends on two main factors:
i) The distance between the nodes.
ii) The infra-structure of the link (between the two nodes).
While the location of the nodes is often fixed and cannot be changed, the infra-structure 
type and service can be upgraded - it is a price-dependent service.

Quality of service is related to different parameters like, bandwidth, delay time, jitter, packet error rate and many others.
Given a network graph, the desired objective is to have the best {\em QoS} for a given (fixed) budget. 
In this paper we focus on minimizing the maximum and the average delay time using a fixed budget.

\subsection{Related Work}
The problems of Center Point, Median Point on graphs (networks) have been studied extensively, see \cite{NDL_D_95,LTUA_D_05} for a detailed surveys on 
{\em facility location}. There are various optimization problems dealing with finding the best graph;
A typical {\em graph or network improvement} problem considers a graph which needs 
to be improved by adding the smallest number of edges in order to satisfy some constraint (e.g., maximal radius), 
see \cite{1115584,ChepoiNV03,OMCBDFBMDEAP-92,966446}. 
Spanner graph problems \cite{GSN_NS_07} consider what can be seen as the {\em inverse} case of {\em network improvement problems}. In a typical spanner problem we would like 
to keep the smallest subset of edges from the original graph while maintaining some constraint. See \cite{GSN_NS_07} for a detailed survey on spanners. 
Observe that both {\em network improvement} and {\em spanner graph} problems can be modeled as a discrete version of our suggested new model.

\subsection{Our Contribution}
In this paper we present linear time algorithms for rooted and unrooted budget radius and budget median problems on trees.
We also prove that the general version of the problem is N.P.-hard and devise an $O(\log^2(n))$ approximation algorithm for the budget radius problem on general metric spaces.

\section{Preliminaries}
In this section we introduce basic notations and definitions that are used for describing the suggested budget graph framework. Let $G = \graph{V,E}$ be a graph with some length function $\length:E\mapsto\R^+$. A {\em valid budget allocation} $\B(\cdot)$ to $E$ is a non-negative real function, such that 
$\sum_{e \in E}{\B(e)} = 1$ (here and in the rest of the paper we assume that the total budget $B$ equals 1; this is without loss of generality since an optimal solution with budget of 1 is easily scaled to any budget $B$). Let $E = \{e_1, \dots, e_{\left| E \right|}\}$. We denote $b_i \eqndef \B(e_i)$, and for every $E' \subseteq E$ we denote $\B(E') = \sum_{e_i \in E'}{b_i}$.
Given a valid budget allocation ${\B}$ to $E$, the {\em weight} of an edge  $e\in E$, 
denoted $\weight_{{\B}}(e)$, is a function of $\length(e)$ and $\B(e)$. Throughout this 
paper we consider the case where $\weight_{\B}(e) \eqndef \frac{\length(e)}{\B(e)}$.
\begin{definition}[weighted distance]
Let $u,v\in V$ be two vertices. The {\em weighted distance} between $v$ and $u$, denoted $\weightedDist_{\B}(v,u)$, is the minimum weight over all simple paths between $u$ and $v$. Namely, 
$\weightedDist_{\B}(v,u) \eqndef \min(\{\sum_{e \in P}{\weight_{\B}(e)}: 
P \text{ is a simple path from } v \text{ to } u\})$.
\end{definition}
Table \ref{tab:notations} summarizes the above notations. In the rest of this section we present two specific facility location problems in the budget graph framework. In Section~\ref{sec:BR}, we describe the budget radius problem and in Section~\ref{sec:BM}, we describe the budget median problem.
\begin{table}[htb!]
\begin{center}
\begin{tabular}{|c|c|}
\hline
\textbf{Notation} & \textbf{Explanation} \\
\hline
\textbf{$G = \graph{V,E}$} & a general undirected graph \\&(induced by some metric space)\\
\hline
\textbf{$\length(e)$} & the (a priori) length of an edge $e \in E$.\\
\hline
\textbf{$\B(e)$} & the budget fraction allocated to $e \in E$.\\
\hline
\textbf{$\B$=$\{b_1,...b_{|E|}\}$} & an alternative notation for the function $\B$.\\
\hline
\textbf{$\weight(e)=\frac{\length(e)}{\B(e)}$} & the (budget implied) weight of $e \in E$.\\
\hline
\textbf{$(G,w)=G(\B)$} & the {\em budget-graph} implied by an allocation $\B$\\
\hline
\textbf{$\weightedDist_{\B}(v,u)$} & the  distance between two vertices in $G(\B)$.\\
\hline
\end{tabular}
\caption{Notations that are used throughout the paper to present the new budget graph model.}\label{tab:notations}  
\end{center}
\end{table}

\subsection{The Budget Radius Problem}\label{sec:BR}
We next introduce some definitions and notations to define the setting of the {\em budget radius} problem on graphs. In the first setting we consider, a candidate center node to the graph is given and the goal is to find an optimal budget allocation that minimizes the radius of graph with respect to the resulting distances between nodes in the graph (the weighted radius). In a second setting, a candidate center node is not given and the goal is to find both a center node and a budget allocation that are together optimal with respect to the weighted radius. The formal definitions follow. 

\begin{definition}[rooted weighted radius]
Given a valid budget allocation $\B$ to $E$ and a vertex $r \in V$, the {\em weighted radius} of $G$ with respect to $r$ is defined as $\WR_{\B} (r) = \WR_{\B} (G,r) \eqndef \max_{v\in V}(\weightedDist_{\B}(r,v))$.
\end{definition}  
Given a graph $G = \graph{V,E}$ (induced by some metric) and a node $r \in V$, we next define what it means for an allocation to be optimal with respect to the budget radius problem. We usually omit $G$ from the notation whenever it can be derived from the context.
\begin{definition}[rooted budget radius]
An optimal allocation for $(G,r)$ with respect to the budget radius problem is a valid allocation for which the weighted radius with respect to $r$ is minimized. There may be several optimal allocations. We take an arbitrary optimal allocation and denote it by $\B_r^* = \B_r^*(G)$. We further refer to this allocation as {\em the optimal allocation} for $(G,r)$.

The {\em budget radius} of $G$ with root $r$, denoted $\BR(r) = \BR(G,r)$, is the weighted radius of $G$ with respect to $r$ and $\B_r^*$, i.e.,  $\BR(r) = \WR_{\B_r^*} (r)$.
\end{definition}

We next give the appropriate definitions for the setting where the root is not given as part of the input to the problem. We call this problem the {\em unrooted} budget radius problem. The general setting, where the goal is to find a node in the graph that minimizes the radius is sometimes referred to as the center point problem. 
\begin{definition}[unrooted budget radius]
The {\em budget radius} of a graph $G$  is  the value $\BR = \BR(G) \eqndef \min_{v\in V}{\BR(G,v)}$.  We refer to a pair $(\B^*,r^*)$ as the optimal allocation for $G$ (in the unrooted setting) if $\B^*$ is a valid allocation to $E$ and $r^*\in V$ is the vertex with the smallest corresponding radius, i.e., $\WR_{\B^*} (r^*) = \BR$.
\end{definition}
 
We demonstrate the above definitions using the few examples presented in \figref{1}.
\begin{figure}[bth]
\centering{\includegraphics[scale=0.850]{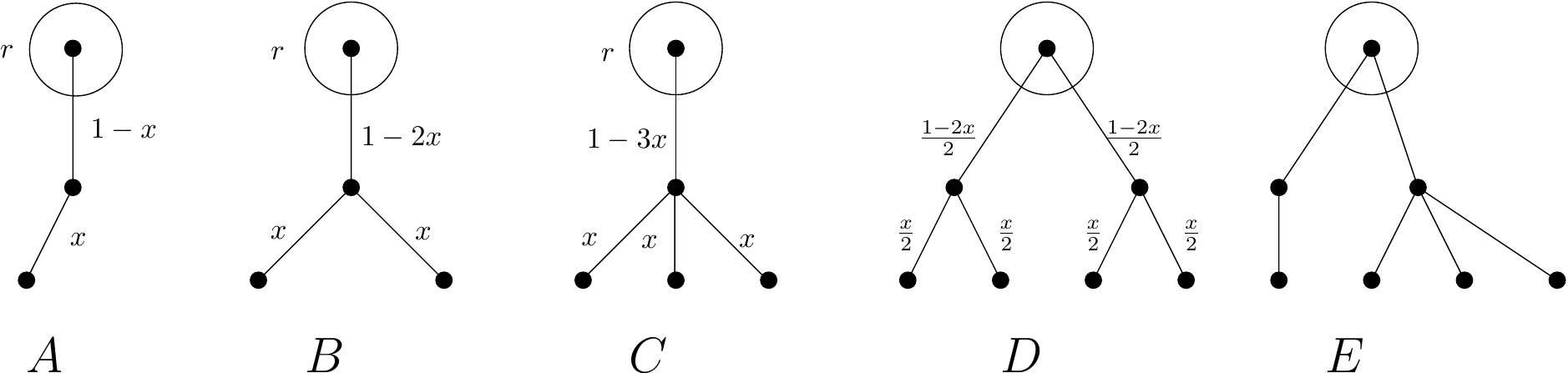}}
\begin{footnotesize}
\caption{Few simple examples of a budget graph radius problem on rooted trees (with a given center $r$). 
Assume that each edge $e$ has length $\length(e)=1$. The total budget is 1 and the optimal budget allocation of each 
edge is shown as a function of $x$. 
Consider case $B$: the optimal solution minimizes the following function: $f(x) = \frac{1}{1-2\cdot x} + \frac{1}{x}$.
Note that in cases where $x$ equals $\frac{1}{3}$ or $\frac{1}{4}$ the radius is 6, while the optimal allocation of $x$ 
is approximately $0.293$, and the radius is approximately $5.828$.
In case $C$ the function is: $f(x) = \frac{1}{1-3\cdot x} + \frac{1}{x}$.
For which $x$ is approximately $0.211$, 
and the radius is approximately $7.464$.
Case $D$ is composed from two cases of $B$ therefore the radius is twice the budget radius of $B = 11.656$.
Case $E$ is composed from cases $A$,$C$ therefore the radius is the sum of the two radii $= 11.464$.}
\label{fig:1}
\end{footnotesize}
\end{figure}

\subsection{The Budget Median Problem}\label{sec:BM}
Given a graph, the general median problem is defined as follows: find a node ($m^*$) in graph from 
which the sum of all weighted distances to all other nodes in the graph is minimized. 
On {\em budget graphs}, the {\em budget median} problem is to find a median node and a corresponding optimal allocation which minimizes the sum of distances between the median node and all the other nodes in the graph (with respect to the budget allocation). 

\begin{definition}[rooted weighted budget median]
Given a valid budget allocation $\B$ to $E$ and a vertex $m \in V$, the {\em weighted median} of $G$ with respect to $m$ is defined as $\WM_{\B} (m) = \WM_{\B} (G,m) \eqndef \sum_{v\in V}(\weightedDist_{\B}(m,v))$.
\end{definition}  

Given a graph $G = \graph{V,E}$ (induced by some metric) and a node $m \in V$, we next define what it means for an allocation to be optimal with respect to the budget median problem. 
\begin{definition}[rooted budget median]
An optimal allocation for $(G,m)$ with respect to the budget median problem is a valid allocation for which the sum of all weighted distances to $m$ is minimized. There may be several optimal allocations. We take an arbitrary optimal allocation and denote it by $\B_m^* = \B_m^*(G)$. We further refer to this allocation as {\em the optimal allocation} for $(G,m)$.

The {\em budget median} of $G$ with root $m$, denoted $\BM(m) = \BM(G,m)$, is the weighted sum of distances
from $m$ to all other nodes of $G$ with respect to $\B_m^*$, i.e.,  $\BM(m) = \WM_{\B_m^*} (m)$.
\end{definition}

We next give the appropriate definitions for the setting where the root is not given as part of the input to the problem. We call this problem the {\em unrooted} budget median problem. 
\begin{definition}[unrooted budget median]
The {\em budget median} of a graph $G$  is  the value $\BM = \BM(G) \eqndef \min_{v\in V}{\BM(G,m)}$.  We refer to a pair $(\B^*,m^*)$ as the optimal allocation for $G$ (in the unrooted setting) if $\B^*$ is a valid allocation to $E$ and $m^*\in V$ is the vertex with the smallest corresponding sum of distances to all other nodes in the graph, i.e., $\WM_{\B^*} (m^*) = \BM$.
\end{definition}


\section{The Budget Radius Problem for Trees}

In this section, we solve the budget radius problem for trees.
We start by showing that the solution to the budget radius problem for a general graph $G$ is always in the form of a tree spanning $G$. In Section~\ref{sec:RootedTree}, we solve the rooted budget radius problem for trees. In Section~\ref{sec:UnRootedTree}, we turn to the unrooted variant of the problem.

\begin{lemma}\label{lem:optimalallocationIsTree}
Let $G=\graph{V,E}$ be a connected graph and let $\length$ be a length function on $E$. An optimal budget allocation (with respect to the budget radius problem) $(B^*,r^*)$ for $G$, has the property that $H = \graph{V,E_{\B^*}}$, where $E_{\B^*}=\{e_i\in E : b^*_i > 0\}$ is a tree spanning $G$. 
\end{lemma}

\begin{proof}
First, assume towards a contradiction that $H$ is not connected. Then, for any center vertex $r^*$, there is a vertex $v\in V$ such that there is no path in $H$ connecting $r^*$ to $v$. In this case, however, any path between $r^*$ and $v$ in $(G,w) =G(\B^*)$ contains an edge of infinite weight. Hence, $\dist_{\B^*}(r,v) =\infty$, and so the radius implied by the budget allocation $\B^*$ is $\infty$. This means that $\B^*$ is not optimal, since there are allocations yielding a finite radius, e.g., the uniform allocation, assigning an equal portion of the budget to each edge (i.e., $\B(e)=\frac{1}{\size{E}}$ for all $e\in E$). 
 
Second, assume towards a contradiction that $H$ contains a cycle. Let $H'$ be the shortest path tree rooted at $r^*$. Hence, there exists at least one  edge $e\in E_{H}\setminus E_{H'}$.  By definition, $\B^*(e) > 0$. We define a new budget allocation $\B'$ distributing the budget portion allocated to $e$ evenly among all edges in $H'$. It is easy to verify that the budget radius of $\B'$ is strictly smaller than that of $\B^*$. This is a contradiction to the optimality of $\B^*$. Hence, $H$ is a spanning tree. 
\end{proof}
An analogous argument applies to the budget median problem.
We note, however, that the above is not true for the budget diameter problem (see figure \ref{fig:1.5}).
\begin{figure}[bth]
\centering{\includegraphics[scale=0.5]{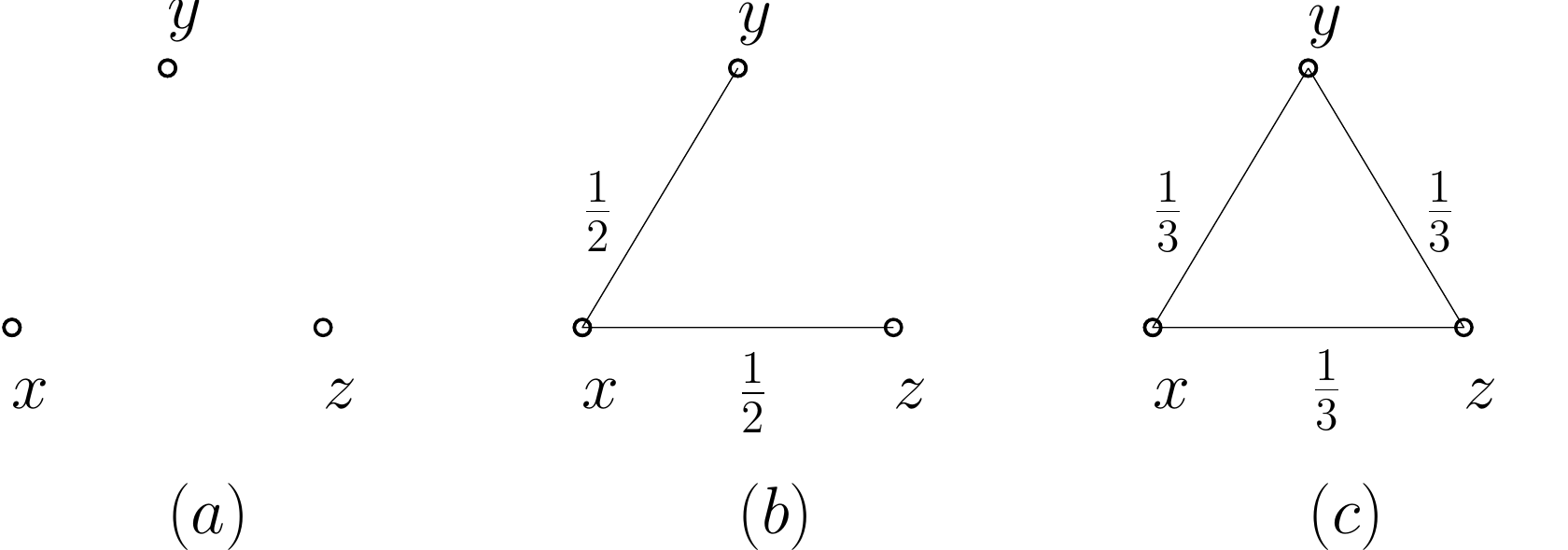}}
\begin{footnotesize}
\caption{ $(a)$ Three points in the plane: $x,y,z$ are the nodes of a unite equilateral triangle. $(b)$ Tree: optimal radius (2), non optimal diameter (4). $(c)$ Cycle graph: non-optimal radius (3), optimal diameter (3).}
\label{fig:1.5}
\end{footnotesize}
\end{figure}
 
The above lemma suggests that it is interesting to consider the Budget Radius problem 
for the subclass of trees. In the sequel, we present an algorithm solving this problem.
We first consider the case where a designated center node $r$ is given as a part of the input, and an optimal budget allocation
$\B^*$ is sought. We use the standard terminology and refer to $r$ as the root of the tree 
(rather than, the center). Thereafter, we consider the general case in trees,
where the problem is to find a pair $(\B^*, r^*)$ minimizing the budget radius of the tree.  

\subsection {The Budget Radius for Rooted Trees}\label{sec:RootedTree}
  We next consider two sub-families of rooted trees that will later be the basis for our recursive construction of an optimal valid budget allocation to the edges of general trees. First, we consider a tree in which the root has just one child.
\begin{lemma} \label{lem:centerOneChildSimpleCase}
   Let $T$ be a tree rooted at $r$, with some length function $\length$ on the edges of $T$. Assume $r$ has a single child $r'$ (the root of the subtree $T'$), and let $R' =\BR(T',r')$ and $q = \length(r,r')$. Then, an optimal budget allocation $\B^*$ assigns to the edge $e = (r,r')$ a fraction $\beta =\frac{\sqrt{q}}{\sqrt{R'}+\sqrt{q}}$.  It follows that $\BR(T,r) = \frac{q}{\beta} + \frac{R'}{1-\beta}$. 
\end {lemma}
%
\begin{proof}
   Let $E$ be the set of edges in $T$ and $E' = E\backslash \set{e}$ be the set of edges of $T'$ (see \figref{2}-a). Given any valid budget allocation $\B$ to $E$, let $\B'$ be the scaling of the restriction of $\B$ to $E'$, defined by $\B'(e') = \frac{\B(e')}{1-\B(e)}$ for every $e' \in E'$. Note that with this scaling, $\B'$ is a valid budget allocation to $E'$, i.e., $\sum_{e' \in E'}{\B'(e')} = 1$. Since any path from $r$ to any leaf of $T$ must start with the edge $e = (r,r')$, it follows that $$\WR_{\B}(r) = \frac{q}{\B(e)} + \frac{\WR_{\B'}(r')}{1-\B(e)}.$$
   Hence, for $\B$ to be optimal for $T$ with root $r$, we must have $\B'$ be optimal for $T'$ and $r'$. In addition, given $R' =\BR(T',r')$, the budget radius of $T$ with root $r$ is obtained by assigning a $\beta$ fraction of the budget to $e$, for $\beta$ that minimizes the function $\WR_{\B}(r) = \frac{q}{\beta} + \frac{R'}{1-\beta}$. It follows that $\BR(T,r) = \frac{q}{\beta} + \frac{R'}{1-\beta}$ for $\beta = \frac{\sqrt{q}}{\sqrt{R'}+\sqrt{q}}$. 
\end{proof}

   We next consider the general family of (rooted) trees (\figref{2}-b). Let $T = (V,E)$ be a tree, rooted at $r$, such that $r$ has $k$ children $r_1,r_2, \dots, r_k$, where $r_i$ is the root of the subtree $T_i = (V_i,E_i)$. Denote by $T'_i = (V'_i,E'_i)$ the subtree of $T$, rooted at $r$ and containing $T_i$. Formally, $V'_i = V_i\cup \{r\}$, and $E'_i = E_i \cup\{(r,r_i)\}$. Clearly, the edges sets $E'_i$s are disjoint. Given a valid budget allocation $\B$ to $E$, for each index $i \in \{1,2,..,k\}$ denote by $l_i \in V_i$ the leaf $l$ for which $\weightedDist_{\B}(r,l)$ is the largest within $T'_i$. 
   Recall that the weighted radius $\WR_{\B}(r)$ is determined by the maximum weighted distance to some $l_i$. In other words, $\WR_{\B}(r) =  \max_{1 \le i \le k}(\weightedDist_{\B}(r,l_i))$. 
   Next, we show that in any optimal budget allocation $\B^*$ for such $T$, the fraction of the budget assigned to the edges of each subtree $T'_i$ is directly correlated 
   to its relative weighted radius.   

\begin{figure}[htb]
\centering{\includegraphics[scale=0.6]{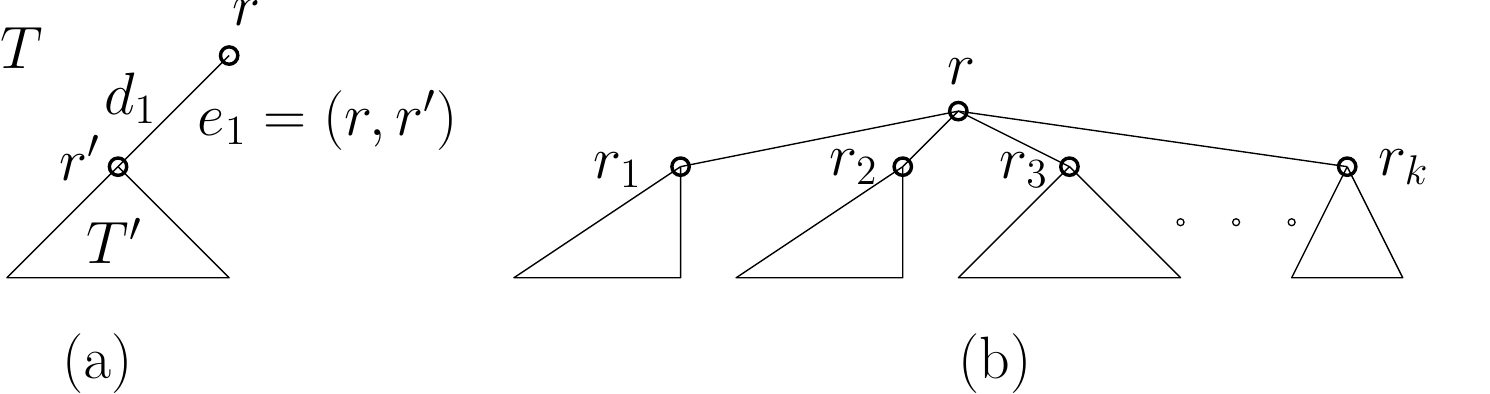}}
\begin{footnotesize}
\caption{(a) The case that $r$ has only a single child. (b) The general case.}
\label{fig:2}
\end{footnotesize}
\end{figure}

   \begin{lemma} \label{lem:centerKChildren}
   	Let $T = (V,E)$ be a tree rooted at $r$, with some length function $\length$ on the edges of $T$. Assume that $r$ has $k$ children $r_1,r_2, \dots, r_k$ where $r_i$ is the root of the subtree $T_i = (V_i,E_i)$, and let $\B^*$ be an optimal budget allocation to $E$. For each $1 \le i \le k$, let $T'_i$ and $l_i$ be as in the foregoing discussion (i.e., 	$l_i$ is maximal in $T'_i$ with respect to $\weightedDist_{\B^*}(r,\cdot)$). Then, for all $1 \le i,j \le k$ that $\weightedDist_{\B^*}(r,l_i)=\weightedDist_{\B^*}(r,l_j)$.
  \end {lemma}

  \begin {proof} 
	  Assume that for some $1 \le i,j \le k$, it holds that 
	  $\weightedDist_{\B^*}(r,l_i) > \weightedDist_{\B^*}(r,l_j)$. We show that it is then
	  possible to present a better budget allocation for $T$, which, in turn, leads to a contradiction. 
	  Let $\rho = \frac{\weightedDist_{\B^*}(r,l_j)}{\weightedDist_{\B^*}(r,l_i)}$ and consider
	  an alternative budget allocation in which each edge $e$ in $E'_j$ was assigned a $\rho$ fraction of its current budget, i.e. $\rho \cdot \B^*(e)$ (while assignment to all other edges stays the same as before). The length of each path from $r$ to a leaf in $T'_j$ would be multiplied by 
	  $1/\rho$. Hence, the maximum distance from $r$ to any leaf in the $T'_j$ would be at most
	  $\weightedDist_{\B^*}(r,l_i)$. This allocation is therefore as good as $\B^*$ (with respect to
	  the weighted radius) although the sum of assigned values is not $1$, but rather, 
	  $1-(1-\rho) \cdot \B^*(E'_j)<1$. Turning it into a valid budget allocation by dividing the remaining
	  $(1-\rho)\cdot\B^*(E'_j)$ budget equally among all edges in $E$, we obtain a better valid budget 
	  allocation to $E$. That is, we fix a new allocation $\B'$ by 
	  setting $\B'(e) = \rho \cdot \B^*(e) + \frac{(1-\rho)\cdot\B^*(E'_j)}{\left|E\right|}$ if $e \in E'_j$,
	  and $\B'(e) = \B^*(e) + \frac{(1-\rho)\cdot\B^*(E'_j)}{\left|E\right|}$ otherwise. 
	  The budget radius of $\B'$ is smaller than that of $\B^*$, in contradiction to the optimality of $\B^*$.
	\end{proof}
	
	The following corollary describes how any optimal valid budget allocation must divide the 
	budget among the disjoint sets of edges of the subtrees $T'_i$.

\begin{coro} \label{cor:calculationMultipleChildren}	  
Let $T$ be a tree as above. Then in any optimal budget allocation $\B^*$ to $E$ it holds that $\B^*(E'_i) = \frac{ \BR(T'_i,r)}{\sum_{j=1}^k{ \BR(T'_j,r)}}$. Thus, an optimal solution in this case is given by $\BR(T,r) = \sum_{j=1}^k{\BR(T'_j,r)}$. 
\end{coro}
\begin {proof}
Let $l_i$'s be as above, i.e., $l_i \in V_i$ is the leaf $l$ for which $\weightedDist_{\B^*}(r,l)$ is maximal within $T'_i$.
Denote $\beta_i = \B^*(E'_i)$ and let $\B'_i$ be the scaling of the restriction of $\B^*$ to $E'_i$, defined by $b'_i(e') = \frac{\B^*(e')}{\beta_i}$ for every $e' \in E'_i$. Clearly, each $\B'_i$ is a valid budget allocation
to $T'_i$. We claim that it is also an optimal one. By \lemref{centerKChildren}, for all $1 \le i,j \le k$ it holds that $\weightedDist_{\B^*}(r,l_i)=\weightedDist_{\B^*}(r,l_j)$. If for some $i$ it holds that $\B'_i$ is not optimal 
for $T'_i$, then choose an optimal valid budget allocation $\B''_i$ for $T'_i$ and scale it back to obtain a valid 
budget allocation $\hat{\B}$ by setting $\hat{\B}(e') = {\B''(e')}\cdot{\beta_i}$ for each $e' \in E'_i$, and 
$\hat{\B}(e) = \B^*_i(e)$ for each $e' \notin E'_i$. 
This reduces the distance of the farthest leaf from $r$ within $T'_i$, while
the distance to any leaf outside $T'_i$ stays as with $\B^*$. Specifically, we have  
$\WR_{\hat{\B}} (r) \le \WR_{\B^*} (r)$. However, by \lemref{centerKChildren}, $\hat{\B}$ is not optimal
and hence, $\B^*$ is not optimal either -- contradiction.

By the above it holds for all $1 \le i \le k$ that $\BR(T'_i,r)={\beta_i}\cdot\weightedDist_{\B^*}(r,l_i)$. 
Thus, by \lemref{centerKChildren}, for all $1 \le i,j \le k$ it holds that  $\frac{\BR(T'_i,r)}{\beta_i} = \frac{\BR(T'_j,r)}{\beta_j}$. Since it also holds that $\sum_{i=1}^k{\beta_i}=1$, we have that for 
all $1 \le i \le k$ it holds that 
$\beta_i = 1-\sum_{j\neq i}{\beta_j} = 1-\beta_i\cdot\sum_{j\neq i}{ \frac{\BR(T'_j,r)}{\BR(T'_i,r)}}$. Hence,
$\beta_i = \frac{ \BR(T'_i,r)}{\sum_{j=1}^k{ \BR(T'_j,r)}}$. 
Furthermore, since for any $1 \le i \le k$ we have that $\BR(T,r) = \weightedDist_{\B^*}(r,l_i)$ (specifically,
since $\BR(T,r) = \weightedDist_{\B^*}(r,l_1)= \frac{\BR(T'_1,r)}{\beta_1}$), it follows that $\BR(T,r) = \sum_{j=1}^k{\BR(T'_j,r)}$.
\end{proof} 

\begin{theorem} \label{thm:optimalSolutionRootedTree}
Given a tree $T$ rooted at $r$, it is possible to find an optimal valid budget allocation for $T$ and $r$, in linear time in the size of $T$.
\end{theorem}

\begin {proof}
$T$ is a rooted tree, thus an inductive construction is only natural. First, assume $T$ is a single node $r$. In this case, no budget is needed and $\BR(T,r) = 0$.
Assume $T$ is rooted at $r$, such that $r$ has $k$ children $r_1,r_2, \dots, r_k$. Denote by $T_i$ the subtree of 	 
$T$ rooted at $r_i$, and containing all vertices (and edges) of the subtree rooted at $r_i$ (and only these vertices).
Denote by $T'_i = (V'_i,E'_i)$ the subtree of $T$ rooted at $r$, induced by adding the edge $(r,r_i)$ to $T_i$. Formally, $V'_i = V_i\cup \{r\}$, and $E'_i = E_i \cup\{(r,r_i)\}$. Thus, each $T'_i$ is a rooted tree where the root ($r$) has a single child, and no $T'_i, T'_j$ for $i \neq j$ share
any vertex other than $r$ and $E'_i, E'_j$ are disjoint for all $i \neq j$. 

By \corref{calculationMultipleChildren}, if we know $\BR(T'_i,r)$ for all
$1 \le i \le k$, we can derive an optimal valid budget allocation for $T,r$. 
In order to obtain a $\BR(T'_i,r)$, it suffices to have an optimal solution for the subtree of $r_i$, 
which, by the induction hypothesis can be done (using \lemref{centerOneChildSimpleCase}). 

Note, that we evaluate the optimal solution for every subtree of every vertex in $T$ exactly once and thus the procedure requires in linear time.
\end {proof}  

 The following lemma proves helpful in the sequel, but is interesting in its own right. It captures some of the tricky nature of the budget radius problem, as it shows the connection between two seemingly unrelated quantities. 
 The first is the weight of a minimum spanning tree ($\MST$) of a given graph and the second is the optimal solution for the budget radius problem for that graph.
 
\begin{lemma} \label{lem:budgetLargerThanSum}
 Given a tree $T=(V,E)$ rooted at $r$, with some length function $\length$ on $E$,
 the budget radius of $T$ is at least the sum of lengths of the edges of $T$, i.e.,
 $\BR(T,r) \ge \sum_{e \in E}{\length(e)}$.  
\end{lemma}

\begin{proof}
 We prove the above lemma by induction. If $T$ has no edges, then both values are $0$.
 If $r$ has only one child $r'$ (the root of the subtree $T'$), then by \lemref{centerOneChildSimpleCase}, 
 since any optimal allocation $\B^*$ must assign $\B^*((r,r'))>0$, we have  
 $\BR(T,r) > \length((r,r'))+\BR(T',r')$, which, by the induction hypothesis is at least 
 $\sum_{e \in E}{\length(e)}$.
 Otherwise, assume $r$ has $k$ children $(r_1 \dots r_k)$ and denote $T_i$ the subtree 
 induced by $r$ and the vertices of the subtree of $r_i$. By \corref {calculationMultipleChildren} 
 $\BR(T,r) = \sum_{i=1}^{k} \BR(T_i,r)$. Hence, by the induction hypothesis, the lemma follows.
\end{proof}

\subsection {The Budget Radius for Unrooted Trees}\label{sec:UnRootedTree}
In this section we consider the budget radius problem for unrooted trees, i.e., where the root of the tree is not given as part of the input. 
Clearly, one can invoke the algorithm from \thmref{optimalSolutionRootedTree} with every vertex $v$ as a candidate center vertex $r$, and select the vertex $v$ for which $\BR(T,v)$ is minimal as the ultimate center. However, this naive algorithm requires $O(n^2)$ time. We next show how to construct a linear time algorithm for this problem (indeed, for a tree $T$, our algorithm computes $\BR(T,v)$ for every $v$ in $T$). Intuitively, this protocol uses the fact that given $\BR(T,r)$ and the partial computations made by algorithm of \thmref{optimalSolutionRootedTree}, applied to the $T$ and $r$, it possible to compute in constant time $\BR(T,v)$ for every neighbor $v$ of $r$.
This intuition is formalized in~\lemref{constTimeComputeBR}.

\begin{lemma} \label{lem:constTimeComputeBR}
Let $T = (V,E)$ be a tree rooted at $r$, with some length function $\length$ on $E$.
Let $v\in V$ be a neighbor (a child) of $r$. Denote by $T_v=(V_v,E_v)$ the subtree of $v$, and denote by $T'_v$ the subtree of $v$ augmented by the edge $e = (r,v)$ (i.e.,
$T'_v=(V_v\cup r, E_v\cup e)$). It is possible to compute, in constant time, $\BR(T,v)$ given $\BR(T,r)$, $\BR(T_v,v)$, and $\BR(T'_v,r)$, see Figure \ref{fig:3}.
\end {lemma}

\begin {proof} 
Denote by $\hat{T}$ the tree obtained by omitting $T_v$ from $T$, formally, 
$\hat{T} = (\hat{V},\hat{E})$, where 
$\hat{V} = V\backslash (V_v \backslash \set{v})$ and 
$\hat{E} = E \backslash E_v$. In addition, denote by $\hat{T'}$ the tree obtained by   
omitting the edge $e=(r,v)$ from $\hat{T}$, i.e., 
$\hat{T}' = (\hat{V}\backslash \set{v},\hat{E}\backslash \set{e})$.

It can be easily derived from \corref{calculationMultipleChildren} that $\BR(T,v) = \BR(T_v,v)+\BR(\hat{T},v)$. 
By \lemref{centerOneChildSimpleCase}, we can compute $\BR(\hat{T},v)$
from $\BR(\hat{T}',r)$ and $\length(e)$, in constant time. Finally, we compute $\BR(\hat{T}',v)$, 
using \corref{calculationMultipleChildren} again, to obtain $\BR(\hat{T'},r) = \BR(T,r) - \BR(T'_v,r)$.
\end{proof}

\begin{figure}[bth]
\centering{\includegraphics[scale=0.40]{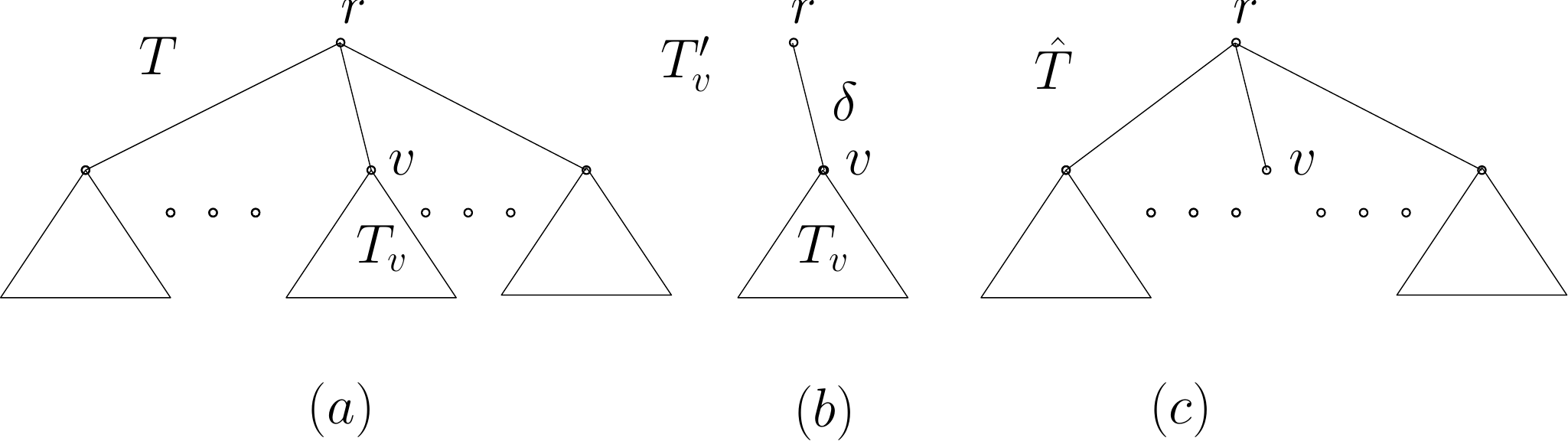}}
\begin{footnotesize}
\caption{(a) The original tree rooted at $r$. (b) Considering $v$ as the root of $T'_v$. (c) The tree $\hat{T}$.}
\label{fig:3}
\end{footnotesize}
\end{figure}
  
\begin{figure}[bth]
\centering{\includegraphics[scale=0.40]{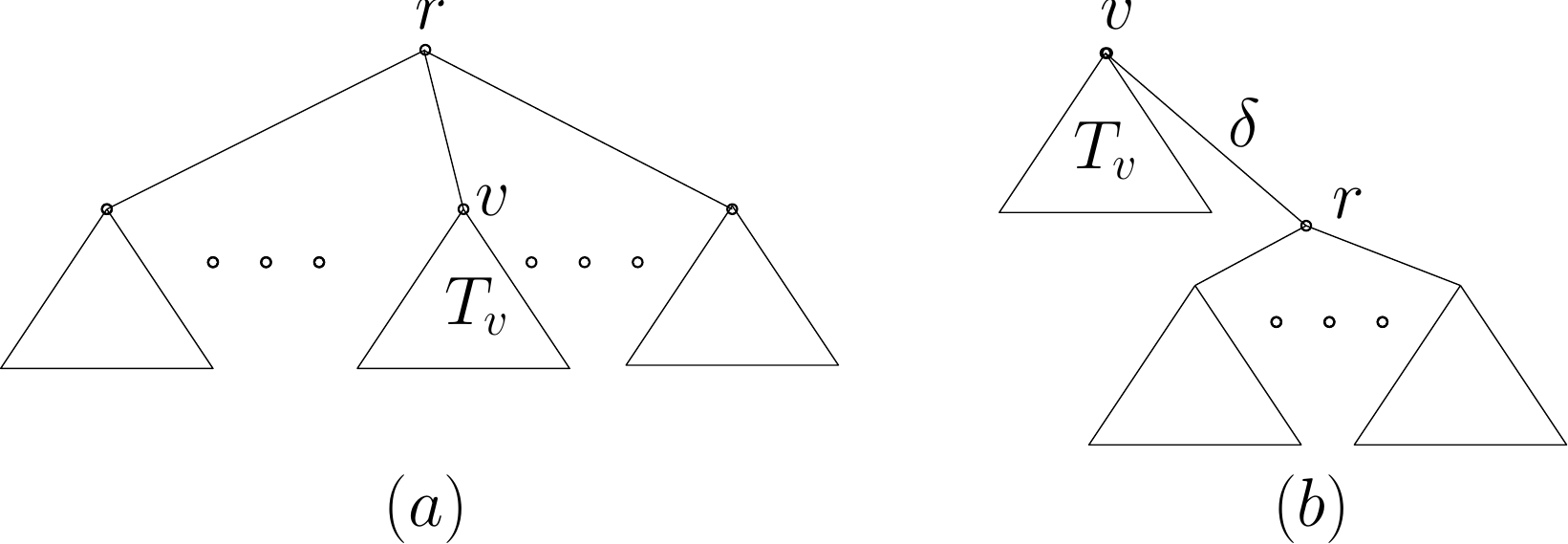}}
\begin{footnotesize}
\caption{(a) The original tree rooted at $r$. (b) Considering $v$ as the root of the tree. Computation of an optimal budget allocation for the tree rooted at $v$ can be done in constant time, given an optimal budget allocation for the tree rooted at $r$.}
\end{footnotesize}
\label{fig:4}
\end{figure}

Roughly, our algorithm will traverse the tree twice. First, we traverse the tree, computing the algorithm 
of \thmref{optimalSolutionRootedTree} for an arbitrary root $r$ (say, $r=v_1$). Recall that this algorithm 
traverses the tree in a bottom-up fashion, i.e., from the leaves to the root, and that an optimal solution for each vertex is calculated, with respect to the subtree below it. Thereafter, we traverse the tree in a top-down fashion, while for each vertex $v$ that is a child of $v'$, we compute an optimal budget radius for the tree with root $v$, given an optimal budget radius for the tree with root $v'$ and the information stored in $v$ from the first traversal. 
\old{
The suggested algorithm runs in linear time and not only computes the optimal budget radius over 
all nodes, it can also be used to when wishing to compute the optimal budget radius for each node 
in the tree as a root.
We next give a more formal description of this procedure. 
}

\begin{theorem} \label{thm:optimalSolutionUnRootedTree}
    Given a tree $T=(V,E)$ with some length function $\length$ on $E$, it is possible 
    to compute an optimal allocation for $T$, i.e., a pair $(\B^*,r^*)$, such that  
	$\WR_{\B^*} (r^*) = \BR(T)$. Furthermore, this can be done in linear time 
	in the size of $T$.
  \end{theorem}  
\begin{proof}
Our algorithm traverses the tree twice. In the first pass, we set an arbitrary vertex $r$ to be the root (say, $r = v_1$) and traverse the tree in a bottom up manner, following the algorithm described in \thmref{optimalSolutionRootedTree}.  

For any vertex $v \in V$, denote by $T_v= (V_v,E_v)$ the subtree of $v$, and denote by $p(v)$ the parent of $v$. Denote by $T'_v$ the subtree of $v$ augmented by the edge $(p(v),v)$, i.e., $T'_v= (V_v\cup \set{p(v)},E_v\cup \set{(p(v),v)})$. We compute for each vertex $v$ the value of an optimal budget radius with respect to the subtree of $v$, i.e., $\BR(T_v,v)$. Recall that in order to do so, we compute for each child $u$ of $v$, not only the budget radius for $T_u$ (i.e., $\BR(T_u,u)$), but also the budget radius for the augmented subtree of $u$ (with root $v$), i.e., $\BR(T'_u,v)$. Here, we also store the two local values, $\BR(T_v,v)$ and $\BR(T_{p(v)},p(v))$, for each node $v$ we traverse.

In the second pass we traverse the tree in a top-down manner, starting at the root $r$ and moving from each vertex to all its children. For each vertex $v$, we compute the budget radius for the whole tree $T$ with root $v$ (i.e., $\BR(T,v)$). By \lemref{constTimeComputeBR}, we can do so in constant time since we have
already computed $\BR(T,p(v))$, as well as $\BR(T_v,v)$ and $\BR(T_{p(v)},p(v))$.
   The algorithm returns the pair $(\B^*,r^*)$, for which, the budget radius is
   minimal. 
\end{proof}

\section{Generalization to Median Point}
In this section we generalize the center point algorithms to the {\em median point}.
In this case we would like to find the node ($M$) and its corresponding budget allocation which minimizes the average (or sum) weight of all shortest path from $M$ to all the graph nodes. 
Following the same framework as in the center point we now would like to find the optimal allocation for $\B^*$, for a fixed node $r$ as a median point.

\begin{figure}[htb] 
\centering{\includegraphics[scale=0.6]{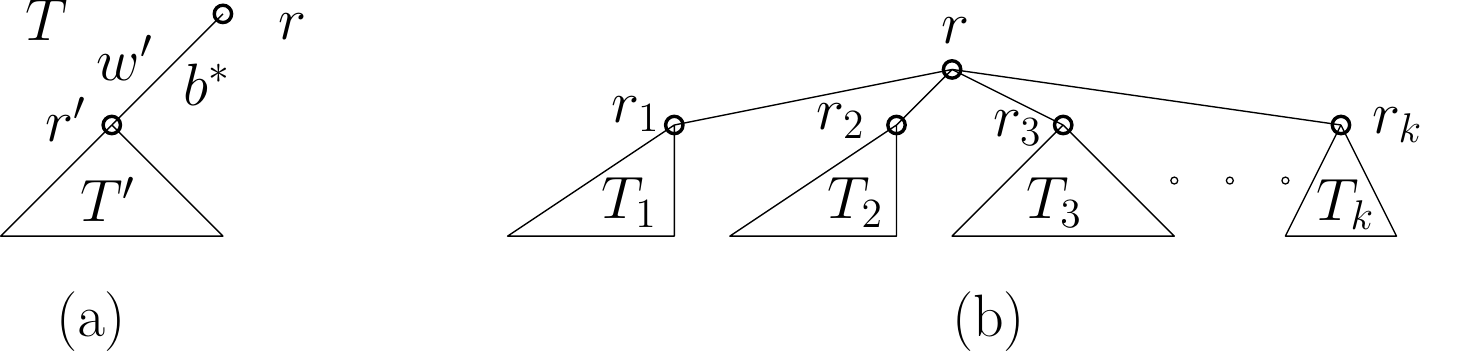}}
\begin{footnotesize}
\caption{Following the same framework as in the center point we now would like to find the optimal allocation for $\B^*$.
(a) Single subtree case: $b^*$ minimizes the value of: $\frac{sum(T',r')}{1-b} + n\frac{w'}{b}$, where $n=|T'|$.
(b) Multiple subtrees case: can be solved as an $LP$ minimization problem.}
\label{fig:6}
\end{footnotesize}
\end{figure}

In case there is a single child to the root $r$ (see Figure \ref{fig:6}$a$) we need to minimize the following term: $sum(T,r) = \frac{sum(T',r')}{1-b^*} + n\frac{w'}{b^*}$, where $n=|T'|$.
Observe that $sum(T,r) = \frac{c_1+\frac{c_2}{x}}{1-x}$ where $c_1 = sum(T',r'), c_2=nw'$ therefore $\B^*$ can be computed in constant time, and so is $sum(T,r)$.

If the root $r$ has more than a single subtree (as shown in Figure \ref{fig:6}$b$), we can apply the above computation on each subtree independently, calculating the optimal budget allocation for each edge $e_i=(r_i,r)$, then the budget between all subtrees can be normalized using the following minimization problem: 
$min(\frac{X_1}{B_1}+\frac{X_2}{B_2}+...+\frac{X_k}{B_k})$ where $B_i>0$, $X_i$ is the optimal sum of $T_i$ with the root $r$, and $B_1+B_2+...+B_k=1$. 
This problem can be solved in $O(k)$ time.
\begin{figure}[bth]
\centering{\includegraphics[scale=1.0]{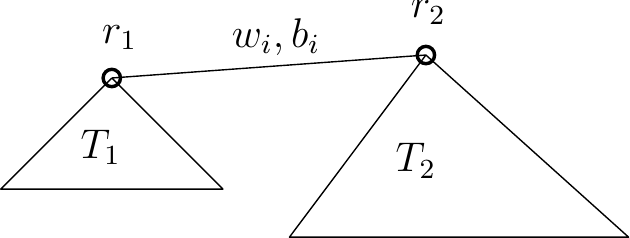}}
\begin{footnotesize}
\caption{Assuming $|T_1| < |T_2|$ the median cannot be in $T_1$, regardless of $w_i$ and $b_i$. }
\label{fig:7}
\end{footnotesize}
\end{figure}

\begin {coro} The location of the median in the unrooted tree case depends only on the tree structure,
i.e., it is in the regular unweighted tree median node. This is due to the convex nature of the median see figure \ref{fig:7}.
Therefore, the optimal root for the median point can be found in linear time.
\end{coro}

\section {Budget Radius -- The General Case}
In this section we consider the general case problem of optimizing the budget radius 
for a complete graph over $n$ vertices, induced by some metric space $M = (V,d)$. 
We present a hardness proof showing that the general case of the budget radius problem is N.P. hard.
We then present an $O(\log^2(n))$ approximation algorithm for this problem.
In order to motivate the non-triviality of a logarithmic approximation factor, 
We start by showing that a naive Minimum Spanning Tree ($\MST$) heuristic may lead to an $O(n^{0.5}$) approximation factor.
Assume we have $n$ points on a square uniform grid. Its $MST$ may have a path like shape, with $\Omega(n)$ radius. 
Hence its budget radius is $\Omega(n^2)$. On the other hand, each of the $n$ points may be connected to the center with a 
path of length $O(n^{0.5})$. 
Hence, the budget radius of this metric is $O(n^{1.5})$.

\begin{figure}[bth] 
\centering{\includegraphics[scale=0.5]{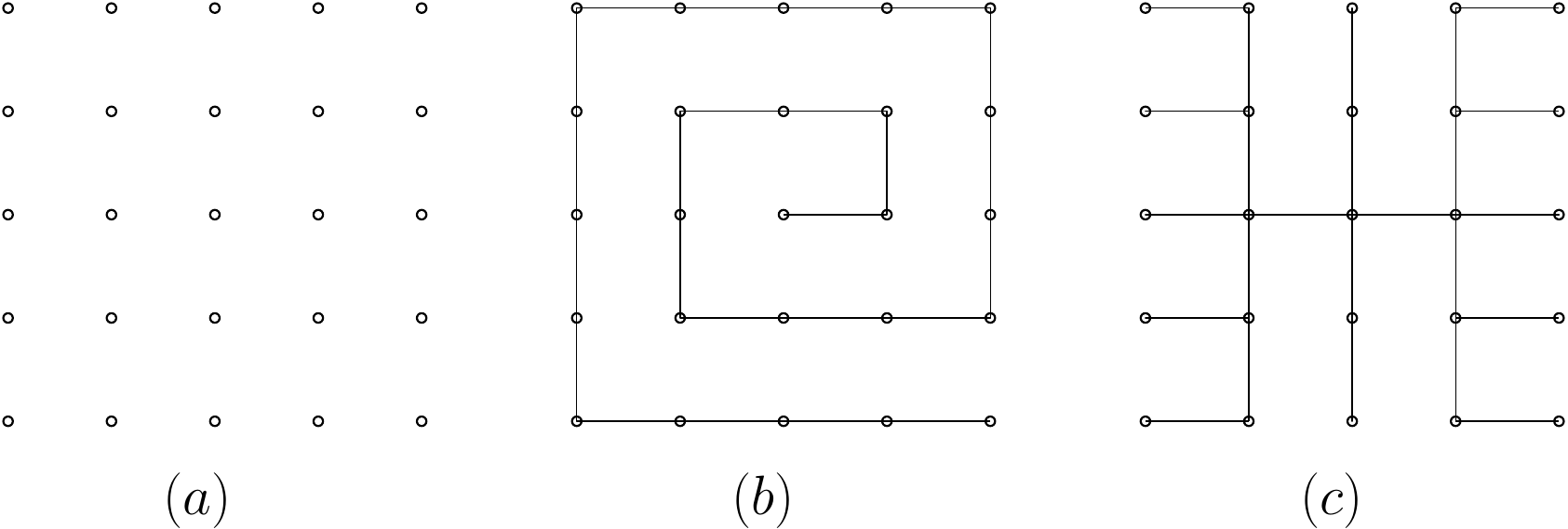}}
\begin{footnotesize}
\caption{Using an $\MST$-like heuristic may lead to an $O(n^{0.5})$ approximation ratio with respect to the {\em center point} problem (minimum radius). ($a$) A grid based set of points. ($b$) A path-like $\MST$. 
($c$) A solution with radius $O(n^{0.5})$.
}
\end{footnotesize}
\label{fig:5}
\end{figure}  

\subsection{Hardness results}
In this section we present a reduction showing that finding the optimal budget allocation 
for radius budget problem on rooted graphs is N.P. hard problem. 
Let $G =(V,E)$ be an undirected graph induced by some length function $\ell(e)$ for each $e \in E$. Let $B$ 
be a positive budget value. Allocating a budget $B(e)$ to edge $e \in E$ with length $\ell(e)$ implies that 
the resulting weight of $e$ is $w(e) = \frac{\ell(e)}{B(e)}$. Let $(G, d_w)$ denote graph $G$ equipped with
a distance function $d_w(\cdot,\cdot)$, defined as the the minimum path distance under
weight function $w(\cdot)$. The rooted budget radius problem is defined as follows: 
Given vertex $r \in V$, allocate budget $B$ among the edges of $E$ (i.e., $\sum_{e \in E}B(e) = B$) so that 
the set radius $\max_{v \in V} d_w(r,v)$ is minimized.
We prove the following theorem:

\begin{theorem}\label{thm:hardness}
The decision version of the rooted budget radius problem is NP-hard.
\end{theorem}

Before proving Theorem \ref{thm:hardness}, we require a preliminary lemma and technical observation:

\begin{lemma}\label{lem:radius}
Let $G=(V,E)$ be a star graph with $V=\{r,s,v_1,\ldots,v_k\}$.
Let $s$ be the center of the star, and let edge
$e=(r,s)$ have length $\ell(e)=1$
and the $k$ edges 
$e_i = (s,v_i)$ 
(for $i=1,\ldots,k$)
have length $\ell(e_i)=x$.
Then in the optimal solution to the rooted budget radius problem on $G$ 
the radius equals 
$$\frac{(1+\sqrt{xk})^2}{B}$$
\end{lemma}

\begin{proof}
First note that in the optimal solution, the path from $r$ to each $v_i$
is equal to $w(e)+w(e_i)$. Hence, it must be that all $k$ allocations 
$B(e_i)$ are equal, as otherwise not all $w(e_i)$ are equal, and then 
some weight could be taken from the largest
allocation $B(e_j)$ (for some $j \in [1,k]$) and distributed evenly among 
all allocations $B(e_i)$, resulting in a shorter radius.

Set $y = B(e)$, and it must be that $B(e_i) = \frac{B-y}{k}$.
It follows that the optimal solution minimizes
$$w(e) + w(e_i) = \frac{1}{y} + \frac{x}{(B-y)/k}.$$
This term is minimized when
$\frac{1}{y^2} = \frac{xk}{(B-y)^2}$,
that is
$y=\frac{B}{1+\sqrt{xk}}$.
So 
$B(e) 
= y 
= \frac{B}{1+\sqrt{xk}}
$
and 
$w(e) 
= \frac{1}{b(e)} 
= \frac{1+\sqrt{xk}}{B}$,
while
$B(e_i) 
= \frac{B-y}{k}
= \frac{B}{k} \cdot \frac{\sqrt{xk}}{1+\sqrt{xk}}
$
and
$w(e_i) 
= \frac{x}{b(e_i)}
= \frac{\sqrt{xk}(1+\sqrt{xk})}{B} 
$.
The radius is 
$w(e) + w(e_i) = 
\frac{(1+\sqrt{xk})^2}{B}
$.
\end{proof}

\begin{observation}
Consider Figure \ref{fig:Hardness1}, The optimal allocations for the left tree (a) and the right tree (b) are not the same, 
the right tree has a smaller budget radius. In fact, the optimal allocation of tree (a) can be used
on the right tree (b) leading to the same radius (allocating to $b_3=b_4=b_1$ and $b_5=b_6=b_2$). The optimal budget allocation
for both trees is presented in figure \ref{fig:1} cases:$D,E$).

\begin{figure}[bth]
\centering{\includegraphics[scale=0.90]{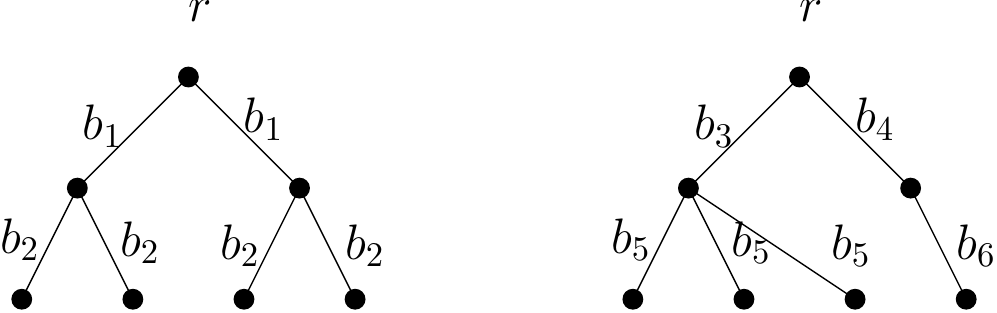}}
\begin{footnotesize}
\caption{The radius budget of the right tree (b) is smaller than the left tree (a).}
\label{fig:Hardness1}
\end{footnotesize}
\end{figure}

\end{observation}
We now return to the proof of Theorem \ref{thm:hardness}:

\begin{proof}
The theorem is proved via reduction to $3$-Set Cover, the version of Set Cover where the
cardinality of each set is at most $3$. $3$-Set Cover is NP-hard. In the reduction, we
will assume that in the minimum radius problem the target radius is fixed and the goal
is to find the minimum cost $B$ -- by scaling, this is equivalent to the scenario where $B$
is fixed as the goal is to minimize the radius.

Given a set cover instance $(S,E)$ of sets and elements, we create a graph $G$ as follows.
First create root node $n_r$.
For each element $e_i$ in the set cover instance, create a single node $n_{e_i}$.
We call these nodes {\em element-nodes}.
For each $3$-Set $S_j = \{e_a,e_b,e_c\}$, we create seven nodes called {\em set-nodes}.
Set 
$x=1$,
$y= \frac{(\sqrt{6}-1)^2}{2}$
and
$z= \frac{(\sqrt{8}-1)^2}{3}$:
\begin{enumerate}
\item
A new set-node $n_{S_j,e_a}$ is connected to root $n_r$ by an edge of length 1, and to
element-node $n_{e_a}$ by an edge of length $x$. New set-nodes $n_{S_j,e_b}$ and 
$n_{S_j,e_c}$ are similarly connected to $n_r$, and respectively to element-nodes $n_{e_b}$ 
and $n_{e_b}$. These set-nodes will represent the solution case where set $S_j$ is needed to 
cover only a single element.
\item
A new set-node $n_{S_j,e_a,e_b}$ is connected to root $n_r$ by an edge of length 1, and to
element-nodes $n_{e_a}$ and $n_{e_b}$ by edges of length $y$. New set-nodes 
$n_{S_j,e_a,e_c}$ and $n_{S_j,e_b,e_c}$ are similarly connected to $n_r$, and respectively to 
element-node pairs $n_{e_a}$ and $n_{e_c}$ or $n_{e_b}$ and $n_{e_c}$. 
These set-nodes will represent the solution case where set $S_j$ is needed to cover
only two elements.
\item
A new set-node $n_{S_j,e_a,e_b,e_c}$ is connected to root $n_r$ by an edge of length 1, and
to each of element-nodes $n_{e_a},n_{e_b},n_{e_c}$ by an edge of length $z$.
This set-node represents the solution case where set $S_j$ is needed to cover all three 
elements.
\end{enumerate}
Let the set of set-nodes for $S_j$ be $N_j$.

Let the radius in the optimal solution be 1, and we will show that finding the minimum 
cost $B$ achieving this radius on $G$ is equivalent to solving set cover on the input 
instance. 

In the optimal solution, for each element-node $n_{e_a}$ only one edge incident on 
$n_{e_a}$ is assigned non-zero weight: At least one edge must be assigned weight, or 
else the path from the root $n_r$ to $n_{e_a}$ is infinite. Further, if the optimal
solution assigns weight to two edges incident on $n_{e_a}$, then it must be that 
one of the incident set-nodes has its minimum path to $n_r$ routed through $n_{e_a}$.
But this forms a contradiction: The edge connecting this set-node to $n_r$ has length
$1 = x < y < z$, so a shorter path can be attained by removing the weight
from the edge connecting $n_{e_a}$ to the set-node, and placing it on the edge 
connecting the set-node to $n_r$. 

It follows that in the optimal solution, the only finite paths are those consisting of 
a single edge connecting the root to a set-node, and the edges connecting the set-node
to at most three element-nodes. So the the set of finite edges form a set of star graphs 
connected at $n_r$.

Now consider a set $N_j$ of set-nodes for set $S_j$, and all edges incident on this set.
Suppose there are exactly $c$ element-nodes whose minimum paths are routed through set-nodes 
of $N_j$. Clearly $0 \le c \le 3$.
\begin{enumerate}
\item
If $c=0$, then weight is assigned only to the length $1$ edges connecting the seven set-nodes 
of $N_j$ directly to the root, at a total cost of $7$. 
\item
If $c=1$, then since $x<y<z$, weight is assigned to the length $x$ edge connecting the 
relevant element-node $n_{e_a}$ to set-node $n_{S_j,e_a}$. Consider the star graph formed by 
nodes $n_r,n_{S_j,e_a},n_{e_a}$: By Lemma \ref{lem:radius}, to achieve a path length of 1 the 
total weight assigned to the two graph edges is $(1+\sqrt{x})^2=4$. The other six edges 
connecting $n_r$ to $N_j$ are each assigned weight $1$, for a total cost of 
$6+(1+\sqrt{x})^2 = 10$.
\item
If $c=2$, the optimal solution will assign weight to the length $y$ edges connecting the relevant
element-nodes $n_{e_a},n_{e_b}$ to set-node $n_{S_j,e_a,e_b}$. Consider the star graph 
formed by nodes $n_r,n_{S_j,e_a,e_b},n_{e_a},n_{e_b}$: By Lemma \ref{lem:radius}, to achieve 
path lengths of 1 the total weight assigned to the three graph edges is $(1+\sqrt{2y})^2=6$. 
The other six edges connecting $n_r$ to $N_j$ are each assigned weight $1$, for a total cost 
of $6+(1+\sqrt{2y})^2 = 12$.

Another option would have been to assign weight to the two edges connecting 
$n_{e_a}$ to $n_{S_j,e_a}$ and $n_{e_b}$ to $n_{S_j,e_b}$. In this case, the total cost 
is $5+2(1+\sqrt{x})^2 = 13 > 12$, so this assignment is suboptimal.
\item
If $c=3$, the optimal solution will assign weight to the length $z$ edges connecting
element-nodes $n_{e_a},n_{e_b},n_{e_c}$ to set-node $n_{S_j,e_a,e_c}$. Consider the star 
graph formed by nodes $n_r,n_{S_j,e_a,e_b,e_c},n_{e_a},n_{e_b},n_{e_c}$: By Lemma 
\ref{lem:radius}, to achieve path lengths of 1 the total weight assigned to the three graph 
edges is $(1+\sqrt{3z})^2=8$. The other six edges connecting
$n_r$ to $N_j$ are each assigned weight $1$, for a cost of 
$6+(1+\sqrt{3z})^2=14$.

Another option would have been to assign weight to three edges connecting 
$n_{e_a}$ to $n_{S_j,e_a}$, $n_{e_b}$ to $n_{S_j,e_b}$, and $n_{e_c}$ to $n_{S_j,e_c}$. 
In this case, the total cost is $4+3(1+\sqrt{x})^2 = 16> 14$, so this 
assignment is suboptimal.
If we would assign weight to the edges connecting $n_{e_a}$ to $n_{S_j,e_a}$ and both 
$n_{e_b},n_{e_c}$ to set-node $n_{S_j,e_b,e_c}$, then the total weight would be 
$5 + (1+\sqrt{x})^2 + (1+\sqrt{2y})^2 = 15 > 14$, so this assignment is suboptimal.
\end{enumerate}

Recall that in the optimal solution, each element-node is incident on a single edge
with non-zero weight, and this edge connects the element-node to a set-node. We will
say that the set-node {\em covers} this element node. As above, for a given set $N_j$, 
at most one set-node in $N_j$ covers some element-nodes. In this case we say that $N_j$ 
is {\em used}, and that $N_j$ covers these element-nodes. 

We may then view the above cost assignment as follows:
For each set $N_j$, if the $N_j$ is not used we pay a single overhead cost of 7, 
and if $N_j$ is used we pay a overhead cost of 8 plus a cost of 2 for each
element-node covered by $N_j$. Then the total cost for the graph is
$7|S|+2|E|$ plus the number of used sets, and the minimum cost is achieved
by using the minimum number of set-nodes to cover all element-nodes.
This is equivalent to solving the Set Cover problem on the input.
\end{proof}

\subsection{Approximation Algorithm - General Case of Budget Radius}
\subsubsection{The Special Case of a Line}
  We first consider a setup in which $M$ is defined by some $n$ points all 
  residing on the interval $[0,1]$, where for any two points $p_1,p_2$ within
  this interval, $d(p_1,p_2)$ is the Euclidean distance between $p_1$ and $p_2$. Let $G =(V,E)$ be the complete graph 
  induced by $M$. We present a valid budget allocation $\B$ to $E$ with 
  budget radius at most $\log^2 n$ and such that the graph induced by 
  $\{e: \B(e)>0\}$ is a tree spanning $V$. 

\begin {lemma} \label{lem:completeGraphOnLineApproximation}
	Let $G = (V,E)$ be the complete graph described above, then $\BR(G) \le \log^2 n$.
\end{lemma}
\begin {proof}
Let $P=\{p_1, p_2,...,p_n\}$ be a set of $n$ points on the interval $[0,1]$ (in increasing order).
Next, we construct a full binary tree $T$ over $P$. Its root is $p_{\frac{n}{2}}$,
and the root's children are $p_{\frac{n}{4}}$ and $p_{{\frac{3}{4}} n}$, etc.

  The set of possible solutions for the budget radius 
  problem for $T$ is a subset of the set of possible solutions for the budget radius 
  problem for $G$. Thus, it suffices to present a solution for $T$ 
  with budget radius at most 
  $\log^2(n)$, that is, a pair of the form $(\B,r')$, where $\B$ is a 
  valid budget allocation to $E_T$ and $r'\in V_t$, such that, 
  $\WR_{\B} (r') \le \log^2 (n)$. Clearly, fixing $r' = r$ only further restricts
  the set of solutions we allow. We next describe one such solution. 
				
  We first divide the edges of $T$ into sets, defined by the depth of a given edge 
  from the root $r$. Formally, if $e$ is an edge in $T$, we say that $e$ is at level 
  $i$ in $T$ if one of its vertices has depth $i$ and the other has depth $i+1$. 
  We denote the set of all edges in $T$ of level $i$ by $S_i$. For every edge 
  $e \in S_i$ we set $\B(e) = \alpha_e$, where 
  $\alpha_e \eqndef \frac{1}{\log (n)} \frac{\length(e)}{\sum_{e' \in {S_i}}{\length(e')}}$.
  We first need to show that this allocation is valid and sums up to at most $1$. 
  This is true since for every level $i$, we divide a $\frac{1}{\log (n)}$ fraction 
  of the budget among the edges in $S_i$. Since there are no more than $\log (n)$ 
  levels, we do not exceed our budget. 
   
  To bound the budget radius of $T$ under this allocation, observe that as $T$ 
  is a search tree, it holds for every $i$ that $\sum_{e' \in {S_i}}{\length(e')} \le 1$. 
  Thus, for every edge $e$ of $T$, we have $\alpha_e \geq \frac{\length(e)}{\log (n)}$. Now, 
  let $P$ be a simple path from $r$ to some leaf $\ell$. The weighted 
  length of $P$ (the weighted distance between $r$ and $\ell$) is
  $\sum_{e \in P}{\weight_{\B}(e)} = \sum_{e \in {P}}{\frac{\length(e)}{\alpha_e}} 
  \le \sum_{e \in {P}}{\log (n)}$, which is at most $\log^2 (n)$ since $P$ consists of at most $\log (n)$ edges.
\end{proof}

\subsubsection{General Complete (Metric) Graphs}
We next define an approximation algorithm $\A$, such that given a 
complete graph $G = (V,E)$, induced by some metric space $M = (V,d)$,
approximates the Budget Radius problem for $G$ by a factor of $O(\log^2 n)$.  
Assume that a minimum spanning tree for $G$ has a total weight $\LB$, we proceed as follows:
\old{
1. Find an Hamiltonian path ($\HP$) visiting all nodes with weight no more than $2 \cdot \LB$.\\
2. Let $G'$ be the result of unfolding $\HP$ to a straight line, i.e., $G'$ is defined by $n$ points, situated on an interval, such that, the distance between every two points is the length of the path between them on the Hamiltonian path $\HP$ (specifically, the length of the whole interval is exactly the length of $\HP$).\\
3. 	Scale the above ($\HP$) interval length to $1$. \\
4. 	Build a balanced binary search tree ($\BT$) over $G'$.\\ 
5. 	Apply the algorithm of \thmref{optimalSolutionRootedTree} to $\BT$. Assign the appropriate budget to all edges in $\BT$ and $0$ to all other edges in $E$.
}
\begin{enumerate}
\item Find an Hamiltonian path ($\HP$) visiting all nodes with weight no more than $2 \cdot \LB$.
\item Let $G'$ be the result of unfolding $\HP$ to a straight line, i.e., $G'$ is defined by $n$ points, situated on an interval, such that, the distance between every two points is the length of the path between them on the Hamiltonian path $\HP$ (specifically, the length of the whole interval is exactly the length of $\HP$).
\item 	Scale the above ($\HP$) interval length to $1$. 
\item 	Build a balanced binary search tree ($\BT$) over $G'$. 
\item 	Apply the algorithm of \thmref{optimalSolutionRootedTree} to $\BT$. Assign the appropriate budget to all edges in $\BT$ and $0$ to all other edges in $E$.
\end{enumerate}

\begin{theorem} \label{thm:approximationSolutionGeneralGraph}
	Let $G = (V,E)$ be a complete graph induced by some metric space $M = (V,d)$. Then, algorithm $\A$ results in a valid budget allocation to the edges of $E$ that approximates $\BR(G)$ by a $2 \log^2(n)$ factor.
\end{theorem}
	
\begin{proof}
 First note that finding an Hamiltonian path ($\HP$) with weight no more than $2\cdot \LB$  is feasible using an $\MST$ for $G$. More importantly, note that by
 \lemref{budgetLargerThanSum}, it holds that $\LB$ is a lower bound
 on the optimal solution (i.e., on $\BR(G)$). This is true since an optimal budget allocation defines a tree (see \lemref{optimalallocationIsTree}), which has a total weight of at least $\LB$ (by the minimality of an $\MST$). Thus, algorithm $\A$ yields an optimal budget allocation for $\BT$, which by
 \lemref{completeGraphOnLineApproximation} yields a budget radius of at most $2 \cdot \LB \cdot \log^2(n) \le 2 \cdot \BR(G) \cdot \log^2(n)$. 
\end{proof}


\section{Conclusion and Future Work}
The paper introduces a new model for optimization problems on graphs. The suggested {\em budget} model was used to define facility location problems such as center and median point. For the tree case, optimal algorithms are presented for both aforementioned problems. For the general metric center point problem, an $O(\log^2(n))$ approximation algorithm is presented.
The new model raises a set of open problems e.g.,: i) Complexity: there is still a considerable gap between the hardness result
and the approximation factor in the general case of budget radius.
ii) Facility location: Find approximation algorithms for the k-center, 
k-median on general graphs.
iii) Graph optimization: minimizing the diameter of the budget graph.
\old{
In this paper we have presented a new model for optimization problems on graphs. This new {\em budget} model was used to define facility location problems such as center and median point. For the tree case, in both problem an optimal algorithms are presented. For the general-graph center point problem a $log^2(n)$ approximation algorithm is presented.
The new model raises a set of open problems e.g.,: i) Hardness: is the budget center point problem on general graphs is NP-hard\footnote{The discrete version of the center point budget problem is weakly NP-complete even on a path like graph (reduced to perfect {\em perfect partition problem}.}.
ii) Facility location: Define approximation algorithms for the k-center, k-median, and 1-median on general graphs.
iii) Graph Optimization: minimizing the diameter of the graph.
}
\bibliographystyle{plain} 
 \bibliography{BG}
 \newpage


 
\end{document}